 \documentclass[11pt]{article}
\usepackage{fullpage}

\usepackage{url}
\usepackage{amssymb,amsmath,amsfonts}
\usepackage{amsmath}
\usepackage{graphicx}
\usepackage{times}
\usepackage{amsthm}
\usepackage{algorithm}
\usepackage[algo2e, ruled, vlined]{algorithm2e}

\newcommand{\ignore}[1]{}

\newcommand{\onlyinproc}[1]{}

\newtheorem{thm}{Theorem}[section]
\newtheorem{theorem}{Theorem}[section]

\newtheorem{lemma}[thm]{Lemma}

\newtheorem{corollary}[thm]{ Corollary}

\newcommand{\dist}{\mathop{\rm dist}}
\newcommand{\CCC}{\mathop\text{{\sc cc}}}
\newcommand{\APS}{\mathop\text{{\sc aps}}}
\newcommand{\CV}{CV}
\newcommand{\W}{\mathop\text{{\sc W}}}

\newcommand{\var}{\mathop{\sf Var}}
\newcommand{\cov}{\mathop{\sf Cov}}
\newcommand{\Prob}{{\mathrm Pr}}
\newcommand{\minmed}{\mathop{\text{\sc MinMed}}}

\title{Average Distance Queries through Weighted Samples in Graphs and Metric
  Spaces: High Scalability with Tight Statistical Guarantees}

 \ignore{
\numberofauthors{1}
\author{
\alignauthor Shiri Chechik, Edith Cohen, Haim Kaplan \\
       \affaddr{Tel Aviv University, Israel}\\
       \affaddr{CA, USA}\\
       \email{edith@cohenwang.com}
}
 }

\author{
Shiri Chechik \thanks{Tel Aviv University, Israel
  {\tt schechik@cs.tau.ac.il},
  {\tt edith@cohenwang.com},  {\tt haimk@post.tau.ac.il}} \and Edith
Cohen$^*$\thanks{Google Research, CA, USA} \and  Haim Kaplan$^*$
}
\begin{document}
\maketitle

\begin{abstract}
The average distance from a node to all other nodes in a graph, or
from a query point in a metric space to a set of points, is a
fundamental quantity in data analysis.   The inverse of the average
distance, known as the (classic) closeness centrality of a node, is a
popular importance measure in the study of social networks.
We develop novel structural
insights on the sparsifiability of the distance relation via weighted
sampling.  Based on that, we present  highly
practical algorithms with strong statistical guarantees for
fundamental problems.
We show that the average distance (and hence the centrality) for all
nodes in a graph can be estimated using $O(\epsilon^{-2})$ single-source distance computations.
For a set $V$ of $n$ points in a metric space, we show that after preprocessing which uses
$O(n)$ distance computations  we can compute a weighted sample
$S\subset V$ of size
$O(\epsilon^{-2})$ such that  the average distance
from any query point $v$ to $V$ can be estimated from the distances
from $v$ to $S$.
Finally, we show that for a set of points $V$ in a
metric space, we can estimate
the average pairwise distance 
using $O(n+\epsilon^{-2})$ distance computations.  The estimate is
based on a weighted sample of $O(\epsilon^{-2})$ pairs of points,
which is computed using $O(n)$ distance computations.
 Our estimates are unbiased with normalized mean square error (NRMSE)
 of at most $\epsilon$.  Increasing the sample size by a
$O(\log n)$ factor ensures that the probability that the relative error exceeds
$\epsilon$ is polynomially small.  
\end{abstract}

\section{Introduction}

Measures of structural centrality  based on shortest-paths distances,
first studied by   Bavelas~\cite{Bavelas:HumanOrg1948},
are  classic tools in the analysis of social networks and other
graph datasets.
One natural measure of the importance of a node in a network is
its {\em classic closeness centrality}, defined as the
inverse of its  average distance to all other nodes.
This centrality measure, which is
also termed {\em Bavelas closeness centrality} or
the {\em Sabidussi Index}~\cite{Freeman:sociometry1977,Freeman:sn1979,wassermansocialnetworksbook}, was proposed by
 Bavelas~\cite{Bavelas:Acous1950}, Beauchamp~\cite{Beauchamp:BS1965},
 and Sabidussi~\cite{Sabidussi:psychometrika1966}.
Formally,  for a graph  $G = (V,E)$ with
$|V|=n$ nodes,  the classic closeness centrality of $v\in V$ is
\begin{equation} \label{Bclosenesscdef}
\CCC(v) = \frac{n-1}{\sum_{u \in V} \dist(u,v)}\ ,
\end{equation}
where $\dist(u,v)$ is the length of a shortest path between $v$ and $u$
in $G$ and $n$ is the number of nodes.  Intuitively, this measure of
centrality reflects the ability of a node to send goods to all
other nodes.

 In metric spaces, the average distance of a point $z$ to a set $V$ of
 $n$ points,  $\sum_{x\in V} \dist(z,x)/n$, is
a fundamental component  in some clustering and classification tasks.  For
 clustering, the quality of a cluster can be measured by the sum of
 distances from a centroid (usually 1-median or the mean in Euclidean data).
  Consequently, the (potential) relevance
of a query point to
 the cluster can be estimated by relating its average distance to the
 cluster points to that of the center or more generally, to the distribution of the average
 distance of each cluster point to all others.  This classification method has
 the advantages of being non-parametric (making no distribution
 assumptions on the data), similarly to the popular
 $k$ nearest neighbors
 \cite{CoverHartkNN:InfoTheory1967} (kNN)  classification.
Average distance based classification
 complements kNN, in that it targets
settings where the
outliers in the labeled points do carry information that should be incorporated in the classifier.
 A recent study \cite{HMY:Neural2012} demonstrated that this is the
 case for  some data  sets in the UCI repository, where
 average distance based classification is much more accurate than
kNN classification.

These notions of centrality and average distance had been extensively
used in the analysis of social networks and metric data sets.  We aim
here to provide better tools to facilitate the computation of these measures on very large data
sets.  In particular, we present estimators with tight statistical guarantees whose computation is highly  scalable.

We consider inputs that are either in the form of an undirected graph (with
nonnegative edge weights) or a set of points in a metric
space.  In case of graphs, distance of  the underlying metric
correspond to lengths of shortest paths.  Our results also extend to
inputs specified as directed strongly connected graphs where the
distance are the round trip distances \cite{classiccloseness:COSN2014}.
We use a unified
notation where $V$ is the set of nodes if the input is a graph, or the set of points in a metric space.
We denote $|V| = n$. We use graph terminology, and mention metric spaces only when there is a difference between the two applications.
  We find it convenient to work with the sum of
distances $$\W(v)= \sum_{u\in V} \dist(v,u)\ .$$
Average distance is then simply $\W(v)/n$ and centrality is
$\CCC(v)=(n-1)/\W(v)$.
Moreover,
 estimates $\hat{\W}(v)$ that are within  a small relative error, that
 is  $(1-\epsilon)\W(u) \leq \hat{\W}(u) \leq  (1+ \epsilon) \W(u)$, imply a
 small relative error on the average distance, by taking
 $\hat{\W}(v)/n$, and for centrality $\CCC(v)$, by taking
$\hat{\CCC}(v) = (n-1)/\hat{\W}(v)$.

 We list the fundamental computational  problems related to these measures.
\begin{itemize}
\item
{\em All-nodes sums}:  Compute $\W(v)$ of all $v\in V$.
\item
{\em Point queries (metric space):}
  Preprocess a set of points $V$ in a metric space, such that given a
  query point $v$ (any point in the metric space, not necessarily
  $v\in V$), we can quickly compute $\W(v)$.
\item
{\em 1-median}: Compute the node $u$ of
maximum centrality or equivalently, minimum $\W(u)$.
\item
{\em All-pairs sum}: Compute the sum of the distances between all
pairs, that is $\APS(V) \equiv \frac{1}{2} \sum_{v\in V} \W(v)$.
\end{itemize}

In metric spaces, we seek algorithms that 
compute distances for a small number of pairs of points.
 In graphs, a distance computation between a specific pair of nodes
 $u,v$ seems to be computationally equivalent in the worst-case to
 computing all distances from a single source node (one of the nodes)
 to all other nodes.
 Therefore,  we  seek algorithms that perform a small number of
 single-source shortest paths (SSSP)
 computations. An SSSP computation in a graph
 can be performed  using Dijkstra's algorithm in time
 that is nearly linear in the number of edges \cite{FreTa87}.
To support parallel
 computation, it is also desirable
 to reduce dependencies between the distance  or single-source
 distance computations.

 The best known exact algorithms for the problems that we listed above do not
 scale well.  To compute $\W(v)$ for all $v$,  all-pairs sum, and 1-median,
  we need to compute the distances between
all pairs of nodes,  which in graphs is equivalent to an all-pairs shortest paths
(APSP) computation. To answer point queries, we need to compute the distances from the query
 point to all points in $V$.  In graphs, the hardness of some of these problems was formalized by
 the notion of subcubic equivalence \cite{WilliamsW:FOCS10}.
Abboud et al \cite{AGV:SODA2015} showed that
  exact 1-median is subcubic equivalent to APSP and therefore is
  unlikely to have a near linear time solution.  We apply a similar
  technique and show  (in
  Section \ref{hardness:sec}) that the all-pairs sum problem is also
   subcubic equivalent to APSP.  In general metric spaces,
exact all pairs sum or 1-median clearly requires
$\Omega(n^2)$ distance computations.\footnote{Take a symmetric distance matrix with all
  entries in $(1-1/n,1]$.  To determine the 1-median we need to
  compute  the exact sum of entries in each raw, that is, to exactly
  evaluate all entries in the raw.  This is because an unread entry of
  $0$ in any raw would determine the 1-median.
Similarly, to compute the exact sum of
  distances we need to evaluate all entries.  Deterministically, this amounts to ${n \choose
  2}$ distance computations.}

  Since exact computation does not  scale to very large data sets,
work in the area focused on  approximations with  small relative errors.
  We measure approximation quality by  the normalized root mean square error
(NRMSE), which is the square root of the expected (over randomization
used in the algorithm) square difference
between the estimate and the actual value, divided by the mean.  When
the estimator is unbiased (as with sample average), this is the ratio between the standard deviation and the mean, which is called the
coefficient of variation (\CV). Chebyshev's inequality implies that the probability that the estimator is within
a relative error of $\eta$ from its mean is at least $1-(\CV)^2/(\eta)^2$. Therefore
a \CV\
 of $\epsilon$ implies that the estimator is within a relative error of $\eta = c\epsilon$ from its mean
with
 probability $\ge 1-1/c^2$.


 The sampling based estimates that we
consider are also well concentrated, meaning roughly that the probability of a larger
error decreases exponentially with sample size. With concentration, by increasing the
sample size by a factor of $O(\log n)$ we  get that the probability
that the relative error exceeds $\epsilon$,  for any one of polynomially many queries, is
polynomially small. In particular, we can
estimate the sum of the distances of the 1-median from all other nodes up to a relative
error of $\epsilon$ with a polynomially small error probability.

\subsection*{Previous work}
We review previous work on scalable approximation of 1-median,
all-nodes sums, and all-pairs sum. These problems were studied in metric spaces and graphs.
A natural approach to approximate the centrality of nodes is to take a uniform sample $S$ of
nodes, perform $|S|$ single source distance computations  to
determine all distances from every $v\in S$ to every $u\in V$, and then estimate
$\W(v)$ by $\hat{\W}(v) = \frac{n}{|S|} \W_S(v)$, where
$\W_S(v) = \sum_{a\in S} \dist(v,a)$ is the sum of the distances from $v$ to the nodes of $S$.
This approach was used by Indyk \cite{Indyk:phd} to compute a $(1+\epsilon)$-approximate
1-median in a metric space using only $O(\epsilon^{-2}n)$ distance
computations (See also \cite{Indyk:stoc1999} for a similar result with
a weaker bound.).  We discuss this uniform sampling approach in more
detail in Section \ref{uniformsamp:sec}, where for completeness, we
show how it can be applied to the all-nodes sums problem.

The sample average of a uniform sample  was also used to estimate all-nodes
centrality~\cite{EW_centrality:SODA2001} (albeit with weaker, additive
guarantees) and to
experimentally identify the (approximate) top $k$ centralities~\cite{OkamotoCL:FAW2008}.   When the distance
distribution is heavy-tailed, however, the sample average as an
estimate of the true average can have a large relative error.  This is
because the sample may miss out on the few far nodes that dominate $\W(v)$.

 Recently, Cohen et al \cite{classiccloseness:COSN2014} obtained $\epsilon$ NRMSE
 estimates for $\W(v)$ for any $v$,
using single-source distance computations from each
 node in a uniform sample of $\epsilon^{-3}$ nodes.
Estimates that are within a relative error of $\epsilon$ for all
 nodes were obtained using $\epsilon^{-3}\log n$ single-source computations. This approach applies in any metric space.
The estimator for a point $v$ is obtained by using
the average of the distances from $v$ to a uniform sample for
  nodes which are ``close'' to $v$ and estimating distances to
nodes ``far'' from $v$ by their distance to
the sampled node closest to $v$.  The resulting estimate is biased,
but obtains small relative errors using essentially the information of
single-source distances from a uniform sample.

For the all-pairs sum problem in metric spaces, Indyk
\cite{Indyk:stoc1999} showed that it can be estimated by scaling up
the average of  $\tilde{O}(n \epsilon^{-3.5})$ distances between pairs
of points selected uniformly at random. The estimate has a relative error of at most $\epsilon$ with constant probability.
Barhum, Goldreich, and Shraibman  \cite{AveDist2007} improved Indyk's
bound and showed that a uniform sample of   $O(n \epsilon^{-2})$
distances suffices and also argued that this sample size is necessary (with uniform sampling).
Barhum et al.\ also showed that in an Euclidean space a similar approximation can be obtained by projecting the points onto $O(1/\epsilon^2)$ random directions and averaging the distances between all pairwise projections.
Goldreich and Ron \cite{GoldreichR08} showed that in an unweighted
graph $O(\epsilon^{-2} \sqrt{n})$  distances between random pairs of points suffice to estimate the sum  of all pairwise distances,   within  a relative error of $\epsilon$, with  constant probability.
They also showed that  $O(\epsilon^{-2}\sqrt{n})$  distances from a fixed node $s$ to random nodes $v$ suffice to estimate $\W(v)$,
  within a relative error of $\epsilon$, with  constant probability.
A difficulty with using this result, however, is that in graphs it is
expensive to compute distances between random pairs of points in a
scalable way:  typically a single distance between a particular pair of nodes $s$ and $t$  is not easier to obtain
than a complete single source shortest path tree from $s$.

\subsection*{Contributions and overview}
  Our design is based on computing a single {\em weighted}  sample that provides
  estimates with statistical guarantees for all nodes/points.   A
sample of size $O(\epsilon^{-2})$ suffices to obtain estimates $\hat{\W}(z)$ with
  a \CV\ of $\epsilon$ for any $z$.  A sample of size $O(\epsilon^{-2} \log n)$ suffices
for ensuring a relative error of at most $\epsilon$ for all nodes in   a graph or for polynomially
  many queries in a metric space, with probability that is at least $1-1/poly(n)$.

  The sampling algorithm is provided in Section~\ref{build:sec}.
 This algorithm computes a {\em coefficient} $\gamma_v$ for each $v\in V$ such that $\sum_v \gamma_v
  = O(1)$. Then for
 a parameter $k$, we obtain sampling probabilities  $p_u \equiv
\min\{1, k \gamma_v\}$ for $u \in V$.
Using the probabilities $p_v$, we can obtain a Poisson sample $S$ of expected size
  $\sum_u p_u = O(k)$ or a VarOpt sample \cite{varopt_full:CDKLT10} that has exactly that size (rounded to an integer).

 We present our estimators in Section \ref{estimators:sec}.
For each node $u$, the inverse probability estimator
$\widehat{\dist}(z,u)$ is equal to  $\dist(z,u)/p_u$
if $u$ is sampled and is $0$ otherwise. Our estimate of the sum $\W(z)$
is the sum of these estimates
\begin{equation} \label{centsampest}
\hat{\W}(z) = \sum_{u\in V} \widehat{\dist}(z,u) = \sum_{u\in S}
\widehat{\dist}(z,u) = \sum_{u\in S} \frac{\dist(z,u)}{p_u}\ .
\end{equation}
Since $p_u>0$ for all $u$, the estimates $\widehat{\dist}(z,u)$ and
hence the estimate $\hat{\W}(z)$ are unbiased.

 We provide a detailed analysis in Section \ref{analysis:sec}.
We will show that our sampling probabilities provide the following guarantees.
When choosing
$k = O(\epsilon^{-2})$,
 $\hat{\W}(z)$ has \CV\ $\epsilon$.
Moreover,
the estimates have good concentration, so using a larger sample size
of $O(\epsilon^{-2} \log n)$ we obtain that the relative error is at
most  $\epsilon$ for all nodes $v\in V$
with probability at least $1-1/poly(n)$.

  In order to obtain a sample with such guarantees for some particular node $z$,
the sampling probability of a node $v$ should be (roughly) proportional to its distance
$\dist(z,v)$ from $z$. 
Such a  Probability Proportional to Size (PPS)
sample of size $k=\epsilon^{-2}$ uses coefficients $\gamma_v =
\dist(v,z)/\W(z)$ and has \CV\ of $\epsilon$.  We will work
with {\em approximate PPS} coefficients, which we define as satisfying
$\gamma_v \geq c \dist(v,z)/\W(z)$ for some constant $c$.  With
approximate PPS we obtain a \CV\ of $\epsilon$ with a sample
of size $O(\epsilon^{-2})$.
It is far from clear apriori, however,  that there is a single set of
{\em universal PPS}   coefficients
which are simultaneously (approximate) PPS  for all
nodes and are of size $\sum_v \gamma_v = O(1)$. That is, a single
sample of size $O(\epsilon^{-2})$, which is independent of $n$ and of the
dimension of the space, would work for all nodes.

  Beyond establishing the existence of universal PPS coefficients,
we are interested in obtaining them, and the sample
  itself, using a near-linear computation.
 The dominant component of the computation of the sampling
 coefficients is  performing $O(1)$  single-source distance computations.
  Therefore, it requires $O(m\log n)$ time in graphs and
 $O(n)$ pairwise distance queries in a metric space.
 A universal PPS sample of any given size $k$ can then be computed in a single
 pass over the coefficients vector $\boldsymbol{\gamma}$ ($O(n)$ computation).  We represent the sample $S$ as a
 collection $\{(u,p_u)\}$ of nodes/points
and their respective sampling probabilities.  We can then use our
sample for estimation using \eqref{centsampest}.

When the input is a graph,  we compute single-source distances from each
node in $S$ to all other nodes in order to estimate $\W(v)$ of all
$v\in V$.  This requires $O(|S|m\log n)$ time and  $O(n)$ space.
\begin{theorem} \label{allnodes:thm}
All-nodes sums ($\W(v)$ for all $v\in V$) can be estimated
unbiasedly as follows:
\begin{itemize}
\item
With \CV\
$\epsilon$, using $O(\epsilon^{-2})$ single source distance
computations.
\item
When using $O(\epsilon^{-2}\log n)$ single source distance
computations, the probability that 
the maximum relative error, over the $n$ nodes, exceeds $\epsilon$ is
polynomially small.  
$$\Pr\left[\max_{z\in V} \frac{|\hat{\W}(z)-\W(z)|}{\W(z)} > \epsilon\right] < 1/poly(n)\ .$$
\end{itemize}
\end{theorem}

 In a metric space, we can estimate $\W(x)$ for an arbitrary query point
 $x$, which is not necessarily a member of $V$, by computing the distances $\dist(x,v)$ for all $v\in S$ and
 applying the estimator \eqref{centsampest}.  Thus, point queries in a
 metric space require $O(n)$ distance computations for preprocessing and
 $O(\epsilon^{-2})$ distance computations per query.
\begin{theorem}  \label{querypoint:thm}
We can preprocess a set of points $V$ in a metric space using $O(n)$
time and $O(n)$
distance computations ($O(1)$ single source distance computations)
to generate a weighted sample $S$ of a desired size $k$.
From the sample, we can unbiasedly estimate $\hat{\W}(z)$
using the distances between $z$ and the points in $S$ with the
following guarantees:
\begin{itemize}
\item
When $k=O(\epsilon^{-2})$, for any point query $z$, $\hat{\W}(z)$ has
\CV\ at most  $\epsilon$.
\item
When $k=O(\epsilon^{-2}\log n)$,
the probability that the relative error of $\hat{\W}(z)$  exceeds $\epsilon$ for
is polynomially small:
$$\Pr\left[\frac{|\hat{\W}(z)-\W(z)|}{\W(z)} > \epsilon\right] < 1/poly(n)\ .$$
\end{itemize}
\end{theorem}

 We can also estimate all-pairs sum, using either primitive of
 single-source distances (for graphs) or distance computations (metric
 spaces).
  \begin{theorem} \label{allpairssum:coro}
All-pairs sum can be estimated unbiasedly with the following
statistical guarantees:
\begin{itemize}
\item
\CV\ of at most $\epsilon$, using
$O(\epsilon^{-2})$ single-source
  distance computations.
With a relative error that exceeds  $\epsilon$ with a polynomially
small probability,  using $O(\epsilon^{-2} \log n)$
 single-source distance computations.
\item
With  \CV\ of at most $\epsilon$, using $O(n+\epsilon^{-2})$ distance
computations. With a  relative error  that exceeds  $\epsilon$ with
polynomially small probability, 
$$\Pr\left[\frac{|\widehat{\APS}(V)-\APS(V)|}{\APS(V)} > \epsilon
\right] \leq 1/poly(n)$$
using $O(n+\epsilon^{-2}\log n)$ distance computations.
\end{itemize}
  \end{theorem}

The proof details are  provided in   Section \ref{allpairssum:sec}.
The part of the claim that uses single-source distance computations is
established by using the
estimate $\widehat{\APS}(V) =\frac{1}{2} \sum_{z\in V} \hat{\W}(z)$.
When the estimates have \CV\ of at most $\epsilon$, even if
correlated, so does the
estimate $\widehat{\APS}(V)$.\footnote{In general if random variables $X$ and $Y$ have CV $\epsilon$ then so does their sum:
$\frac{\var(X+Y)}{(E(X+Y))^2} = \frac{\var(X) + \var(Y) + 2\cov(X,Y)}{(E(X+Y))^2}\le  \frac{\var(X) + \var(Y) + 2\sqrt{\var(X)\var(Y)}}{(E(X+Y))^2} \le\frac{\epsilon^2(E(X))^2 + \epsilon^2(E(Y))^2 + 2\epsilon^2E(X)E(Y)}{(E(X+Y))^2} \le \epsilon^2 $.}
  For the  high probability claim, we use
$O(\log n)$ single-source computations to ensure
we obtain universal PPS coefficients with high
probability (details are provided later), which imply that each
estimate $\hat{W}(z)$, and hence the sum is concentrated.

For the second part that uses distance computations, we consider an
approximate PPS distribution that is with respect to 
 $\dist(u,v)$, that is, the probability of sampling the pair $(u,v)$ is at
 least $c \dist(u,v)/\APS(V)$ for some constant $c$.   We show that we
 can compactly  represent this 
distribution as the outer product of two probability vectors of size $n$. 
Using this representation we can draw $O(\epsilon^{-2})$ pairs
independently in linear time, which we use for estimating the average.


Compared to the all-nodes sums algorithms of \cite{classiccloseness:COSN2014}, our result here improves the dependency in
$\epsilon$ from $\epsilon^{-3}$ to $\epsilon^{-2}$ (which is likely to
be optimal for a sampling based approach), provides an unbiased
estimates, and also facilitates approximate average distance oracles with very
small storage in metric spaces (the approach of
\cite{classiccloseness:COSN2014} would require the oracle to store a
histogram of distances from each of $\epsilon^{-3}$ sampled nodes).
For the all-pairs sum problem in graphs, we obtain an algorithm that
uses $O(\epsilon^{-2})$ single source
distance computations, which improves over an algorithm that does $O(\epsilon^{-3})$ single source
distance computations implied by \cite{classiccloseness:COSN2014}.
For the all pairs sum problem in a metric space, we obtain a CV of $\epsilon$
using $O(n+\epsilon^{-2})$ distance computation rather than $O(n\epsilon^{-2})$ distance computations required by the algorithms in
\cite{AveDist2007,Indyk:stoc1999}.

While our analysis does not optimize constants, our algorithms
are very simple and we expect them to be effective in applications.

\ignore{
\medskip
\noindent
{\bf Uniform sampling}  We show (Section \ref{uniformsamp:sec})  how the basic approach of Indyk and
Thorup can be extended to solve the all centralities problem using
$O(\epsilon^{-2} \log n)$ single source computations.  It can also be
for  CCC oracles in a metric space using preprocessing that amounts to
$\epsilon^{-2} \log n)$ single-source distances computation, and storage
of $\epsilon^{-2} \log n)$.   Note that the factor of $\log n$ is
inherent even if we are interested in per-query guarantees.
}

\section{Constructing the sample} \label{build:sec}

  We present Algorithm \ref{centsamp:alg}  that computes a set of sampling
 probabilities associated with the nodes of an input graph $G$.
 We use graph terminology but the
algorithm applies both in graphs and in metric spaces.   The input to the algorithm is a set $S_0$ of base nodes and a parameter $k$ (we discuss how to choose $S_0$ and $k$ below). The algorithm consists of
the following stages.  We first compute a  sampling
coefficient $\gamma_v$ for each node $v$ such that $\sum_v \gamma_v = O(1)$.   Then we use the  parameter $k$ and compute the
 sampling probabilities $p_v=\min\{1, k\gamma_v\}$.  Finally we use
the probabilities $p_v$ to draw a sample of expected size $O(k)$, by choosing $v$ with probability $p_v$.   We
 usually apply the algorithm once with a pre-specified $k$ to
obtain a sample, but there are applications (see discussion in
Section \ref{adaptice1med:sec}) in which we want to choose the sample size
adaptively using the same coefficients.

\begin{algorithm}[h]
\caption{Compute universal PPS coefficients and sample \label{centsamp:alg}}
\KwIn {Undirected graph with vertex set $V$ or a set of points $V$ in a metric space, base
  nodes $S_0$, parameter $k$}
\KwOut{A universal PPS sample $S$}
\tcp{Compute sampling coefficients $\gamma_v$}
\ForEach{node $v$}{ $\gamma_v\gets 1/n$}
\ForEach{$u \in S_0$}
{
Compute shortest path distances $\dist(u,v)$  from $u$ to all other nodes $v\in V$\\
$W \gets \sum_v \dist(u,v)$\\
\ForEach{node $v\in V$}
{$\gamma_v \gets \max\{\gamma_v, \frac{\dist(u,v)}{W} \}$}
}
\ForEach(\tcp*[h]{Compute sampling probabilities $p_v$}){node $v\in V$}{$p_v \gets \min\{1, k\gamma_v \}$}
 $S \gets \emptyset$ \tcp{Initialize sample}
\ForEach(\tcp*[h]{Poisson sample according to $p_v$}){$v\in V$}
{
  \If{$rand() <p_v$}{$S\gets S \cup\{(v,p_v)\}$}
}
\Return{$S$}
\end{algorithm}

\paragraph*{Running time and sample size}
 The running time of this algorithm on a metric space is dominated
 by $|S_0|n$  distance computations.
On a graph, the running time is $|S_0|m\log n$, and is dominated by the $|S_0|$
single-source shortest-paths computations.
The expected size of the final sample $S$ is  $\sum_v p_v \leq k \sum_v
\gamma_v = O(k)$.

\paragraph*{Choosing the base set $S_0$}
We will show that in order to obtain the property that each
estimate $\hat{\W}(v)$  has \CV\ $O(\epsilon)$, it suffices 
that the base set $S_0$  includes a  
 uniform sample of $\geq 2$  nodes and we
need to choose $k=\epsilon^{-2}$.  Note that the \CV\ is computed
over the randomization in the choice of nodes to $S_0$ and of the
sample we choose using the computed coeffcients.
 We will also introduce a notion of a {\em well positioned} node,
  which we precisely define in the sequel.  We will see that when
  $S_0$ includes such a node, we also have \CV\ of $O(\epsilon)$ with
  $k=\epsilon^{-2}$.
This time using only the randomization in the selection of the sample.
 Moreover, if we choose $k=\epsilon^{-2} \log n$ and ensure that $S_0$
 contains a well-positioned node with probability at least
 $1-1/poly(n)$ then we obtain that the probability that the relative
 error exceeds $\epsilon$ is polynomially small.
We will see that most nodes are well positioned, and therefore, it is
relatively simple  to identify such a node with high probability.

\section{Estimation} \label{estimators:sec}

\subsection{Centrality values for all nodes in a graph}
For graphs, we compute estimates  $\hat{\W}(v)$ for all nodes $v\in V$ as in Algorithm \ref{estgraphs:alg}.
We initialize all estimates to $0$, and perform a SSSP
computation from each node in $u \in S$.  When scanning node $v$, during such SSSP computation,  we add $\dist(u,v)/p_u$ to the estimate $\hat{\W}(v)$.
The algorithms runs in $O(|S| m\log n)$ time,
dominated by the $|S|$ SSSP computations from each node in the
sample $S$.

\begin{algorithm}[h]
\caption{Compute estimates $\hat{\W}(v)$  for all nodes $v$ in the graph\label{estgraphs:alg}}
\KwIn{Weighted graph $G$, a sample $S=\{(u,p_u)\}$}
\ForEach{$v\in V$}{$\hat{\W}(v) \gets 0$}
\ForEach{$u\in S$}
{
Perform a single-source shortest-paths computation from $u$.\\
\ForEach{scanned node $v\in V$}
{
$\hat{\W}(v) \gets \hat{\W}(v) + \dist(u,v)/p_u$
}
}
\Return{$(v,\hat{\W}(v))$ for $v\in V$}
\end{algorithm}


\subsection{Point queries (metric space)}

 For a query point $z$ (which is not necessarily a member of  $V$), we
compute the distance $\dist(z,x)$ for all $x\in S$, and apply
\eqref{centsampest}.
This takes $|S|$ distance computations for each query.

\section{Correctness} \label{analysis:sec}

  We first show (Section \ref{unibase:sec}) show that when $k=\epsilon^{-2}$, and
 $S_0$ includes either a uniform sample of size at least $2$ 
then  each estimate $\hat{\W}(v)$ has \CV\ of $O(\epsilon)$.
We then define  well-positioned nodes in Section \ref{wpnodes:sec} and show
that if $S_0$ contains a well positioned node we and sample size is
$k=\epsilon^{-2}$
then the \CV\ is $O(\epsilon)$ (Section \ref{basewp:sec}) and when 
  $k=O(\epsilon^{-2}\log n)$, the
  probability that the relative
  error  exceeds $\epsilon$ is polynomially small (Section \ref{hpests:sec}).

 In Section \ref{gammaupper:sec} we establish an interesting property of our sampling
 coefficients:  They can
not grow too much even if the base set $S_0$ is very large.
 Clearly, $\sum_v \gamma_v \leq 1+|S_0|$, but we will
show that it is $O(1)$ regardless of the size of $S_0$.

We start with some useful lemmas.
\begin{lemma}  \label{ppsq:lemma}
Suppose that  $S_0$ contains a node $u$.  Consider a node $z$
such that $u$ is the $(qn)^{th}$ closest node to $z$. Then for all nodes $v$,
\begin{equation} \label{ppsq:eq}
\gamma_v \geq \frac{1-q}{4} \cdot \frac{\dist(z,v)}{\W(z)} \ .
\end{equation}
\end{lemma}
\begin{proof}
From the specification of Algorithm \ref{centsamp:alg}, the sampling
coefficients $\gamma_v$ satisfy
\begin{equation} \label{probgamma:eq}
\gamma_{v} \ge \max\left\{\frac{1}{n}, \frac{\dist(u,v)}{\W(u)} \right\} \ .
\end{equation}
 Let $Q=\dist(z,u)$.
Consider a classification  of the nodes $v\in V$ to ``close'' nodes and ``far'' nodes according to their distance from $z$:
\begin{eqnarray*}
L &=& \{v\in V \mid \dist(z,v)\leq 2Q\} \\
H &=& \{v\in V \mid \dist(z,v)> 2Q \} \ .
\end{eqnarray*}

Since $\gamma_v \geq 1/n$, for $v\in L$ we have
\begin{equation} \label{gammav-L-bound} \gamma_v \ge \frac{1}{n} \ge  \left(\frac{1-q}{2}\right) \left( \frac{2}{1-q}\right) \frac{1}{n} =
 \left(\frac{1-q}{2}\right) \left(\frac{2Q}{(1-q)Q}\right) \frac{1}{n} \ge
 \left(\frac{1-q}{2} \right)  \frac{\dist(z,v)}{\W(z)} ,
\end{equation}
where the last inequality holds since for $v\in L$ we have $\dist(z,v) \le 2Q$, and
since $\W(z) \geq \left(1- q \right) nQ$ if $u$ is the $(qn)$th closest node to $z$.

For all $v$, we have  that $\dist(u,v) \geq \dist(z,v)-Q$ by the triangle inequality.  We also
have $\W(u) \leq \W(z) + nQ$.
Substituting into (\ref{probgamma:eq}) we get that for every $v$
\begin{equation} \label{gamma_vbound}
\gamma_v \geq  \frac{\dist(u,v)}{\W(u)} \geq  \frac{\dist(z,v)-Q }{\W(z) +nQ}\ .
\end{equation}
In particular, for $v\in H$, we have
\begin{equation} \label{d-Qbound}
\dist(z,v)-Q \geq \frac{1}{2} \dist(z,v)\ .
\end{equation}
As already mentioned, we also have $\W(z) \geq \left(1- q \right) nQ$ and thus
\begin{equation} \label{NQbound}
nQ \leq \frac{\W(z)}{1-q} \ ,
\end{equation}
and
\begin{equation} \label{WnQbound}
\W(z)+nQ \leq \W(z)\left(1+\frac{1}{1-q} \right) = W(z) \left( \frac{2-q}{1-q} \right) \ .
\end{equation}
Substituting \eqref{WnQbound} and \eqref{d-Qbound} in \eqref{gamma_vbound}, we obtain that for $v\in H$,
\begin{equation} \label{gamma_vbound1}
\gamma_v \geq \frac{\dist(z,v)-Q }{\W(z) +nQ}  \geq \frac{1}{2}
\left(\frac{1-q}{2-q} \right)  \frac{\dist(z,v)}{\W(z)} \ .
\end{equation}

The lemma now follows from (\ref{gammav-L-bound}) and (\ref{gamma_vbound1}).
%
\end{proof}

\begin{lemma}  \label{varq:lemma}
Consider a set of sampling coefficients $\gamma_v$ such that for a node
$z$, for all $v$ and for some $c>0$, $\gamma_v \geq c \frac{\dist(z,v)}{\W(z)}$.
Let $S$ be a sample obtained with
 probabilities  $p_v = \min\{1,k\gamma_v\}$
(as in Algorithm \ref{centsamp:alg}), and let
$\hat{\W}(z)$ be the inverse probability estimator as in \eqref{centsampest}.
 Then
\begin{equation} \label{qdependence}
\var[\hat{\W}(z)] \leq
 \frac{\W(z)^2}{k\cdot c}  \ .
\end{equation}
\end{lemma}
\begin{proof}
The variance of  our estimator is
\begin{eqnarray} 
\var[\hat{\W}(z)] & = & \sum_v \left[p_v \left(\frac{\dist(z,v)}{p_v}
    - \dist(z,v)\right)^2 + (1-p_v)\dist(z,v)^2\right] \nonumber \\
&=&   \sum_{v} \left(\frac{1}{p_v}-1 \right) \dist(z,v)^2\ . \label{eq:var}
\end{eqnarray}

 Note that nodes $v$ for which $p_v=1$ contribute $0$ to
the variance.  For the other nodes we use the lower bound $p_v \geq c
k \frac{\dist(z,v)}{\W(z)}$.
\begin{eqnarray*}
\sum_{v\in V} \left(\frac{1}{p_v}-1 \right) \dist(z,v)^2   &=&  \sum_{v\in V \mid p_v<1}
\left(\frac{1}{p_v}-1 \right) \dist(z,v)^2  \\
 &\leq&   \frac{\W(z) }{k\cdot c} \sum_{v\in V} \dist(z,v)\\
&\leq&  \frac{\W(z)^2}{k\cdot c}\ .
\end{eqnarray*}
\end{proof}

\subsection{Base set containing a uniform sample} \label{unibase:sec}
We now consider a situation where $S_0$ includes a uniform sample
  of nodes, and consider the corresponding expected approximation quality:
\begin{lemma} \label{var:lemma}
Suppose that  $S_0$ contains a uniform random sample of $b$ nodes.  Then for
any node $z$,
\begin{equation} \label{bdependence}
\var[\hat{\W}(z)] \leq
\frac{\W(z)^2}{k}  \frac{4b}{b-1} \ .
\end{equation}
\end{lemma}
\begin{proof}
We apply Lemma \ref{varq:lemma} with the bound on the coefficients as
in Lemma \ref{ppsq:lemma} with $u$ being the closest node to $z$ in
$S_0$.   Assume that $u$ is the $x$th closest node to $z$. By Lemma \ref{ppsq:lemma} and Lemma \ref{varq:lemma} we have
\begin{equation} \label{varR:eq}
\var[\hat{\W}(z) \mid x] \leq
 \frac{\W(z)^2}{k} \frac{4}{1-x/n}\ .
\end{equation}
Observe that $x$ is a random variable which is the rank  (= position in the sorted order of the nodes by distance
from $z$) of the closest node to $z$ in a
 uniform sample of size $b$. In particular $x$ take values  $\in \{1,2,\ldots ,n-b+1 \}$ ($x=1$ iff $u=z$).
We have that the probability of  rank  $x$ is
$$
b \left(\frac{1}{n} \right) \left(\frac{n-x}{n-1}\right)\left(\frac{n-x-1}{n-2}\right)\cdots \left(\frac{n-x-b+2}{n-b+1}\right) \le
b\left(1-\frac{x}{n} \right)^{b-1} \ .
$$
(We choose the random subset of $S_0$ of $b$ nodes without replacement, we split into $b$ events according to the step in which the node of rank $x$ is chosen. Other items should be chosen from the $n-x$ nodes of rank larger than $x$. )
The variance $\var[\hat{\W}(z)]$ is the expectation, over $x\in \{1,2,\ldots ,n-b+1 \}$, of
$\var[\hat{\W}(z) \mid x]$.
So from \eqref{varR:eq}, we get
\begin{eqnarray*}
\var[\hat{\W}(z)] & \le &  \sum_{x=1}^{n-b+1} b\left(1-\frac{x}{n} \right)^{b-1} \left( \frac{\W(z)^2}{k} \frac{4}{(1-x/n)} \right) \\
&\le & \frac{\W(z)^2}{k} 4b \sum_{x=1}^{n-b+1} \left(1-\frac{x}{n} \right)^{b-2} \\
&\le & \frac{\W(z)^2}{k} 4b \int_0^1 (1-y)^{b-2} dy \\
& = & \frac{\W(z)^2}{k} \frac{4b}{b-1} \ .
\end{eqnarray*}
\end{proof}

 It follows from Lemma \ref{var:lemma} that if we choose $b\geq 2$ nodes uniformly
into $S_0$ and $k=\epsilon^{-2}$,  then for any
 node $z$, our  estimator has
$\var[\hat{\W}(z)] = O(\epsilon^2\W(z)^2)$.
 This concludes the proof of the per-node (per-point) $O(\epsilon)$ bound on the CV
 of the estimator in the first part of Theorems \ref{allnodes:thm} and~\ref{querypoint:thm} for a sample of size $O(\epsilon^{-2})$.

\subsection{Well-positioned nodes} \label{wpnodes:sec}

We provide  a precise definition of
a {\em well positioned}  node.  Let the {\em median distance} of a node
$u$, denote by $m(u)$, be the distance between $u$ and the $\lceil 1+n/2
\rceil$ closest
node to $u$ in $V$. Let $\minmed = \min_{v\in V} m(v)$ be the minimum median
distance of any node $v\in V$. 
In a metric space, we can define $m(u)$ for any point $u$ in 
the space (also for $u\not\in V$), and accordingly, define $\minmed$ as the 
minimum $m(u)$ over all points $u$ in the metric space.

We say that a node
$u$ is {\em well positioned} if  $m(u) \leq 2\minmed$,
that is, $m(u)$, the median distance of $u$ is within a factor of 2 of
the minimum median distance.
 We now show that most nodes are well positioned.
\begin{lemma} \label{mostwellpos:lemma}
Let $v$ be such that is $m(v) = \minmed$. Then  all
 $\lceil 1+n/2 \rceil$ nodes in $V$ that are closest to  $v$  are well
  positioned.
\end{lemma}
\begin{proof}
Let $u$ be one of the
$\lceil 1+n/2 \rceil$ nodes closest to $v$.
Then $\dist(u,v) \le \minmed$ and a ball of radius $2\minmed$ around $u$
 contains all the
$\lceil 1+n/2 \rceil$ nodes closest
to  $v$. So $m(u) \leq 2\minmed$.
\end{proof}

We are interested in well positioned nodes because of the following
property:
\begin{lemma} \label{wellpos:lemma}
If $u$ is a well positioned node, then for every node
$z$ we have that $\dist(z,u) \le 3m(z)$.
\end{lemma}
\begin{proof}
For every two nodes $u$ and $z$ we have that
$\dist(u,z) \le m(u) + m(z)$ since there must be at least
one node $x$ that is both within distance $m(u)$ from $u$ and within
distance $m(z)$ from $z$, and
by the triangle inequality $\dist(u,z) \le \dist(u,x) + \dist(x,z)$.
The lemma follows since
  if $u$ is well positioned then $m(u) \le
2m(z)$.
\end{proof}
As we shall see, this means that sampling probabilities
proportional to the distances from a well positioned node $u$ approximate sampling probabilities
proportional to the distances from any other node $z$, for nodes whose distance from $z$ is substantially larger than $m(z)$.

\subsection{Base set with a well-positioned node} \label{basewp:sec}

We now consider the case where $S_0$ contains a well-positioned
node.  We show that in this case the coefficients $\gamma_v$ satisfy
what we call a {\em universal PPS} property:
\begin{lemma}  \label{ppswp:lemma}
Suppose that  $S_0$ contains a well-positioned node $u$.
Then  for all nodes $v$,
\begin{equation} \label{ppswp:eq}
\gamma_v \geq \frac{1}{18}  \max_z \frac{\dist(z,v)}{\W(z)} \ .
\end{equation}
\end{lemma}
\begin{proof}
We show that for any node $z$,
$\gamma_v \geq \frac{1}{18}  \frac{\dist(z,v)}{\W(z)}$ using a
variation of the proof of Lemma \ref{ppsq:lemma}.

We partition the nodes into two sets.  A set $L$ which contains the
nodes $v$ such that $\dist(z,v) \leq 6 m(z)$ and a set $H$ which contains the
remaining nodes.
By the definition of $m(z)$ we have that  $\W(z) \geq m(z)(\lfloor \frac{n}{2} \rfloor -1) \ge m(z) \frac{n}{3}$ (for $n\ge 9$).
We obtain that for all $v\in L$,
$$\frac{\dist(v,z)}{\W(z)} \leq \frac{6 m(z)}{m(z)\frac{n}{3}} = \frac{18}{n}\ .$$  Therefore,
$$\gamma_v \geq \frac{1}{n} \geq \frac{1}{18} \frac{\dist(v,z)}{\W(z)}\
.$$
We next consider $v\in H$. Since $u$ is well positioned, by Lemma \ref{wellpos:lemma} we have that  $\dist(z,u) \leq 3 m(z)$.  From the triangle inequality,
$\dist(u,v) \geq \dist(z,v)-\dist(z,u) \geq \dist(z,v) -3m(z) \geq
\dist(z,v)/2$. We also have
$\W(u) \leq \W(z) +n\dist(z,u)\leq \W(z)+3n m(z) \leq 9 \W(z)$.
Therefore $$\gamma_v \geq \frac{\dist(u,v)}{\W(u)} \geq
\frac{(\dist(z,v)/2)}{9\W(z)} = \frac{1}{18} \frac{\dist(z,v)}{\W(z)}\ .$$
\end{proof}
As a corollary, applying Lemma \ref{varq:lemma}, we obtain:
\begin{corollary}
If $S_0$ contains a well-positioned
node,  then for any node $z$,
$\var[\hat{\W}(z)] \leq 18 \frac{\W(z)^2}{k}$.
\end{corollary}

\subsection{Upper bound on the sum of the coefficients} \label{gammaupper:sec}

  One consequence of Lemma \ref{ppswp:lemma} is that the
coefficients
  $\gamma_u$ cannot grow too much even if the base set $S_0$ includes
  all nodes.
\begin{corollary} \label{uppercoeff:lemma}
Let
$$\overline{\gamma}_v \equiv
  \max_z\frac{\dist(z,v)}{\W(z)}\ .$$
Then
$$\sum_v \overline{\gamma}_v = O(1)\ .$$
\end{corollary}
\begin{proof}
Consider the case where $S_0$ consists of a single well positioned node.
By the definition  of $\gamma_v$ we have that $\sum_v \gamma_v \le 2$.
By Lemma \ref{ppswp:eq} we have $\gamma_v \ge \frac{1}{18} \max_z \frac{\dist(z,v)}{\W(z)}$. Therefore
 $\sum_v \overline{\gamma}_v \leq 18 \sum_v \gamma_v \leq 36$.
\end{proof}

\subsection{High probability estimates} \label{hpests:sec}
  Lastly, we establish concentration of the estimates, which will
  conclude the proof of the very high probability claims in
  Theorem~\ref{allnodes:thm} and~\ref{querypoint:thm}.

We need the following lemma:
\begin{lemma}  \label{concentration:lemma}
If our sampling coefficients are {\em approximate PPS} for a
node $z$,  that is, there is a
constant $c$ such that for all nodes $v$,
$\gamma_v \geq c \frac{\dist(z,v)}{\W(z)}$, and we use
 $k=O(\epsilon^{-2}\log n)$, then
$$\Pr\left[\frac{|\hat{\W}(z)-\W(z)|}{\W(z)}\geq \epsilon\right] =
O(1/poly(n))\ .$$
\end{lemma}
\begin{proof}
We apply the Chernoff-Hoeffding bound.
Let $\tau = W(z)/(ck)$.
We have
\begin{equation}  \label{pdef:eq}
 p_v \geq \min\{1,
\dist(z,v)/\tau\} = \min\{1, c k \dist(z,v)/\W(z)\}\ .
\end{equation}

 The contribution of a node $v$ to the estimate $\hat{\W}(z)$ is as follows.
  If  $\dist(z,v) \geq \tau$, then the contribution is exactly
  $\dist(z,v)$.
Otherwise, the contribution $X_v$ of node $v$ is $\dist(z,v)/p_v \leq \tau$ with
probability $p_v$ and $0$ otherwise.

The contributions $X_v$ of the nodes
with $\dist(z,v) \leq  \tau$ are thus  independent  random variables, each in the range
$[0,\tau]$ with expectation $\dist(z,v)$.
  We complete the proof by applying the Chernoff-Hoeffding bound to
  bound the deviation of expectation of the sum of these random variables.   We defer the details to the full version of the paper.
\end{proof}

We need the
condition of Lemma \ref{concentration:lemma}  to hold for  all nodes $z$ with
probability $1-O(1/poly(n))$.  Equivalently, we would like
$\boldsymbol{\gamma}$ to be universal
PPS with very high probability.  If so, we apply a 
union bound to obtain that the estimates $\hat{\W}(z)$ for all nodes $z$ have a relative error of at most
$\epsilon$ with probability $1-O(1/poly(n))$.  The same argument applies to polynomially many queries in
metric spaces.

It follows from Lemma \ref{ppswp:lemma} that we obtain the universal
PPS property if $S_0$ includes a well positioned node.
We would like this to happen with 
very high probability. We mention several ways to achieve this effect:
(i) Since  most nodes are well
positioned (Lemma \ref{mostwellpos:lemma}), taking
a uniform random sample $U$ of $O(\log n)$ nodes, 
and choosing the node $u=\arg\min_{u\in U} m(u)$ with
minimum distance to its $\lceil n/2+1 \rceil$ closest node, means that
we are guaranteed with probability $1-1/poly(n)$
 that $u$ is well positioned.  This identification
step involves $O(\log n)$ single-source distance computations.
(ii) Alternatively, we can ensure that $S_0$ 
  contains a well positioned node (with a polynomially small error) by simply
  placing  $O(\log n)$ uniformly selected
 nodes in $S_0$.  The computation of the coefficients will then
 require $O(\log n)$ single-source distance computations.
(iii) Lastly, if $S_0$ contains $O(\log n)$ uniformly selected
 nodes then we can apply a direct argument that with a polynomially
 small error for each node $z$, one of the $\lceil n/2+1\rceil$ closest nodes to $z$
is in $S_0$. This means we can
 apply Lemma \ref{ppsq:lemma} with $q \le 0.5$ to obtain that with
a polynomially small error, the sampling probabilities are approximate PPS for
 all nodes and thus universal PPS with a polynomially small error.

 To establish the second part of Theorem~\ref{querypoint:thm} in
 metric spaces, we
 would like to identify a well positioned node with a polynomially small
($O(1/poly(n))$) error  using only $O(n)$ distance computations, which
is more efficiently than by using $O(\log n)$ 
 single-source distance computations.

To do so, we first provide a slightly relaxed definition of well positioned node
and show that it retains the important properties.   We will then show
that a ``relaxed'' well positioned node can be identified with very
high probability using only $O(\log^2 n)$ distance computations. 
When we identify such a node, we can use it in the base set $S_0$.  This
means we can use $O(n)$ distance
computations in total to compute coefficients $\boldsymbol{\gamma}$
which are universal PPS with a polynomially small error.  We then use
$O(n)$ time to compute a sample of size $k=O(\epsilon^{-2}\log n)$,
and use this sample to process point queries.

What remains is to introduce the relaxed definition of a
well-positioned node and show that it has the claimed properties.

\subsection{Relaxed well positioned points} \label{relwellpos:sec}
For $Q\geq \lceil 1+ n/2\rceil$, we define the $Q$-quantile distance
$m_Q(v)$ for a point $v$ as the distance of the $Q$th closest point to
$v$.  We then define $\minmed_Q$ as the minimum $Q$-quantile distance over
all points. Now, we define a point $v$ to be $Q$ well positioned if
$m_{\lceil 1+ n/2\rceil}(v) \leq 2 \minmed_{Q}$.

Now observe that at least half the points have $m_Q(v) \leq
2\minmed_{Q}$ and in particular are well positioned (extension of
Lemma \ref{mostwellpos:lemma}).  Also observe
that if $z$ is $Q$ well positioned then for
any node $u$, $\dist(z,u) \leq 3m_Q(u)$ (extension of Lemma \ref{wellpos:lemma}).
We can also verify that for any $Q< 0.6n$ (any constant strictly smaller than
$1$ would do),  a base set $S_0$ containing one $Q$ well positioned
point would also yield coefficients that satisfy the  universal PPS
 property,  albeit with a slightly larger constant.

 We next show that we can identify a $0.6n$ well positioned point
  within a polynomially small error using very few distance computations:
\begin{lemma} \label{metricwellpositioned:lemma}
We can identify a $0.6n$ well positioned point with probability $1-O(1/poly(n))$
using $O(\log^2 n)$ distance computations.
\end{lemma}
\begin{proof}
We  choose uniformly at random  a set of points $C$ of size $O(\log
n)$.  For each point in $v\in C$, we choose a uniform sample $S_v$ of $O(\log n)$
 points and compute the 0.55 quantile of $\{\dist(v,u) \mid  u\in S_v\}$.  We then
 return the point $v\in C$  with the minimum sample 0.55 quantile.

We refer to $C$ as the set of candidates.  Note that since at least
half the points $v\in V$ are such that
$m_{0.6n}(v) \leq 2 \minmed_{0.6n}$, the set $C$
contains such a point with probability $1-O(1/poly(n))$.

 The estimates are such that with probability $1-O(1/poly(n))$, for
 all points in $C$, the sample $0.55$ quantile is between the actual
 $0.5$ and $0.6$ quantiles.  Therefore the point we returned (with a
 polynomially small error)  has $m_{0.5n}$ at most the smallest $m_{0.6n}$ in
 $C$, which is at most $2 \minmed_{0.6n}$.
\end{proof}

\section{All-pairs sum} \label{allpairssum:sec}
  We now establish the claims of Theorem \ref{allpairssum:coro} for
  the all-pairs sum problem.  We start with the first  part of the claim, which is
  useful for graphs, estimates
  $\APS(V)$  using single-source computations.
To do so, we apply Algorithm \ref{centsamp:alg} to compute sampling coefficients
$\boldsymbol{\gamma}$ and then apply Algorithm \ref{estgraphs:alg} to compute
estimates $\hat{\W}(v)$ for all
  $v$.    Finally, we return  the estimate
  $\widehat{\APS}(V) =\frac{1}{2} \sum_{z\in V} \hat{\W}(z)$.

 To obtain an estimate $\widehat{\APS}(V)$ with \CV\ of at most
$\epsilon$, we choose a base set $S_0$ that contains  $2$ uniformly
sampled nodes when applying Algorithm \ref{centsamp:alg}.
We then use sample size of $O(\epsilon^{-2})$ to ensure that the
per-node  estimates $\hat{\W}(z)$ have
\CV\ of at most $\epsilon$.  Note that the estimates of different
nodes are correlated, as they all use the same sample, but
the \CV\ of the sum of estimates each with \CV\ of at most $\epsilon$ must be
at most   $\epsilon$.  The total time amounts to $O(\epsilon^{-2})$
single-source distance computations.

To obtain universal PPS with polynomially small error we can identify a well positioned
node with a polynomially small error, which can be done using $O(\log n)$
single-source computations.  We then compute the sampling coefficients
$\boldsymbol{\gamma}$ for a base set that contains this
well-positioned node. (Which uses a single-source distance computation).
The sampling coefficients we obtain
have the universal PPS property  and the sample-based estimates are
concentrated.  A sample size of size $O(\epsilon^{-2}\log n)$ would yield a relative
error of at most $\epsilon$ with probability $1-1/poly(n)$, for each
$\hat{\W}(z)$ and thus for the sum $\widehat{\APS}(V)$.  In total, we
used
$O(\epsilon^{-2}\log n)$ single-source computations.

 The remaining part of this section treats the second part of the claim of
  Theorem \ref{allpairssum:coro}, which applies to the all-pairs sum
  problem in metric spaces.
We start with an overview of our approach.
In order to obtain a good sample of pairs, we would like to sample
pairs proportionally to $p_{ij} =
\frac{\dist(i,j)}{\APS(V)}$.
  The obvious
difficulty we have to overcome is that the explicit computation of the probabilities $p_{ij}$
requires a quadratic number of distance calculations.

  Our first key observation is that
 we can obtain a sample with (nearly) the same
 statistical guarantees if we relax a little the sampling
 probabilities and the sample size:
 For  some constant $c \geq 1$, we work with probabilities that
 satisfy $p_{ij} \geq c^{-1} \frac{\dist(i,j)}{\APS(V)}$ and use a
 sample of size $k=c \epsilon^{-2}$.

 We use independent sampling with replacement to compute a multiset
 $S$ of pairs of points from $V\times V$.
 The estimator we use is
the sample average inverse probability estimator:
$$\widehat{\APS}(V) = \frac{1}{|S|} \sum_{(i,j)\in S}
\dist(i,j)/p_{ij}\ .$$
This sample average is an unbiased estimate of $\APS(V)$ and has CV of at
most $\sqrt{k/c}$ which is $\epsilon$ when we use sample size
$k=c\epsilon^{-2}$.
Moreover, each summand is by definition at most
$c\APS(V)$ and therefore we obtain concentration by a direct
application of Hoeffding's inequality: The probability of a
relative error that is larger than $\epsilon$ when the sample size is
$k$ is at most $2e^{-2k\epsilon^{2} c^{-2}}$.
In particular, if we take a sample size that is $O(\epsilon^{-2}\log
n)$, we obtain that the probability that the relative error exceeds
$\epsilon$ is polynomially small in $n$.

 We next discuss how we facilitate such sampling efficiently.
 We would like to be able to
sample with respect to relaxed $p_{ij}$ and also have the sampling
probabilities available for estimation.
We show that we can express a set of relaxed
probabilities (for some constant $c$) as the outer product of
two probability distributions over points, $\boldsymbol{\gamma} \boldsymbol{\rho}^T$.
The distribution $\boldsymbol{\gamma}$  has the universal PPS
property with respect to some constant $c'$.  The probability distribution
  $\boldsymbol{\rho}$ has the property that for some constant $c''$, for all $v$,
 $\rho_v\geq c'' \frac{\W(v)}{\sum_u \W(u)}$.  We now observe that
for some constant $c = c' c''$, for all pairs $u,v$,
$\rho_u  \gamma_v \geq c  \frac{\dist(u,v)}{\APS(V)}$.  That is,  we can sample according to $p_{uv} = \rho_u  \gamma_v$
and satisfy the relaxed conditions and obtain the desired statistical guarantees.

  What remains is to provide details on (i) how we use the vectors $\boldsymbol{\gamma}$ and $\boldsymbol{\rho}$
  to obtain a sample of pairs and (ii) how we compute such vectors that satisfy
  our conditions within a polynomially small error.  These are addressed in
  the next two subsections.

\subsection{Sampling pairs using the coefficient vectors}

We show how we obtain $k$ samples $(v,u)$ from $\gamma_v\rho_u$
efficiently, using computation that is  $O(n + k)$.
Many sampling schemes (with or
without replacement) will have the concentration properties we seek
and the implementations are fairly standard.
For completeness, we describe a scheme that computes independent
samples with replacement.   Our scheme obtains a sample from $V\times
V$ by sampling independently a point $i$ according to the probability
distributions $\boldsymbol{\gamma}$ and  a point $j$ according to distribution
$\rho$ and returning $(i,j)$.

 What remains is to describe how we can obtain $k$ independent samples
 with replacement from a probability vector $\boldsymbol{\gamma}$ in time $O(n+k)$.

We arbitrarily order the points, WLOG $i\in V $ is the $i$th point in
the order. We compute  $a_i =
\sum_{h<i} \gamma_h$ and associate the intervals $[a_i,a_i+\gamma_i]$
with the point $i$.

To randomly draw a point $i\in V$  according to $\boldsymbol{\gamma}$, we can draw a random
number $x \sim U[0,1]$ and take the point $i\in V$ such that $x\in
[a_i,a_i+\gamma_i)$.   If we have $k$ sorted random values, we can map
all of them to points in $V$ in $O(n)$ time using one pass on the sorted values and
the sorted nodes.  For completeness, we describe one way to obtain a
sorted set of $k$ independent random draws $x_1,\ldots,x_k \sim
U[0,1]$  using $O(k)$ operations: (i) We draw $k$ values $y_1,\ldots y_k$ where
$y_i \sim Exp[k+1-i]$ is  exponentially
distributed with parameters $k+1-i$.  This can be done by drawing
independent uniform $u_i \sim U[0,1]$ and take $y_i =
-\ln(u_i)/(k+1-i)$.
(iii) Now observe that $x'_i \equiv \sum_{j\leq i} y_j$ for $i\in [k]$
are $k$ independent exponential random
variables with parameter $1$ which are sorted in increasing order.  We
can then
transform $x'_i$ to uniform random
variables $x_i$ using $x_i = 1-\exp(-x'_i)$.  Since the
transformation is monotone, we obtain that $x_i$ are sorted.
 Note that  prefix sums of $y_j$ and hence all $x_i$ can be computed in $O(k)$
 operations.
Also note that we only need precision to the point needed to identify
the point that each $x_i$ maps into.

\subsection{Computing the coefficient vectors}

  We recall that universal PPS coefficients can be computed using
  Algorithm~\ref{centsamp:alg} using $n$ distance computations (and $O(n)$
  additional computation), when our base set $S_0$ contains a well
  positioned point.  The probability vector $\boldsymbol{\gamma}$ we
 work with is the universal PPS  coefficients scaled to have a
  sum of $1$.

  We next discuss how we obtain the probability distribution $\boldsymbol{\rho}$.
  We show that given a $0.6n$ well positioned point (see Section ~\ref{relwellpos:sec}), we can
  compute $\rho_v$ that has the claimed properties with very  high
  probability.
From Lemma \ref{metricwellpositioned:lemma}, we can identify a point
that is $0.6n$  well positioned with probability at least $(1-1/poly(n))$, using only  $O(\log^2 n)$ distance computations.
    We use the following lemma, which a variation of claim used for  the pivoting upper
    bound estimate in \cite{classiccloseness:COSN2014}.  What it
    roughly says is that for any node $u$ and any node $z$ that is within a constant times some
    quantile distance from $u$, we can get a constant factor
    approximation of $\W(u)$ from $\W(z)$ and $\dist(u,z)$.
\begin{lemma}
Consider a point $u$ and a point $z$ such that $\dist(u,z)$
is at most $c$ times the distance of the $(qn)^{th}$ closest point to $u$.
Then
$$\W(u) \leq n \dist(u,z) + \W(z) \leq \left( 1+ \frac{2c}{1-q}\right)  \W(u)\ .$$
\end{lemma}
\begin{proof}
Left hand side is immediate from the triangle inequality.  To establish the right hand
side, first note that
$(1-q)n$ of the points are at least as far as $\dist(z,u)/c$, thus
$W(u) \geq \frac{(1-q)}{c} n \dist(u,z)$.
From triangle inequality we have $\W(z)\leq \W(u)+n\dist(u,z)$.
Combining we get:
$$\W(z) + n\dist(u,z) \leq  \W(u) + 2n\dist(u,z) \leq
(1+\frac{2c}{1-q}) \W(u)\ .$$
\end{proof}

 Now consider a point $z$ that is $0.6n$ well positioned
and using the rough estimates
$$\hat{{\W}'}(u) = n \dist(u,z) + \W(z)$$ for all points $u$ and
accordingly the sampling probabilities
$$\rho_i = \frac{\hat{{\W}'}(i)}{\sum_j \hat{\W'}(j)}\ .$$

By definition, for all points $u$, the point $z$ satisfies
$\dist(u,z) \leq 3m_{0.6n}(u)$. We therefore can apply the lemma with
$q=0.6$ and $c=3$ and obtain that for all $v$,
$\rho_v \geq \frac{1-q}{1-q+2c} \frac{\W(v)}{\sum_j \W(j)}$.
Note that given $z$, the vector $\boldsymbol{\rho}$ can be computed for all points
using $n$ distance computations, from $z$ to all other points.

\section{Uniform sampling based estimates} \label{uniformsamp:sec}

  For completeness, we briefly present another solution for the
  all-points/nodes problem that is based on uniform sampling.  The
  disadvantages over our weighted sampling approach is that it provides
  biased estimates and requires $\epsilon^{-2}\log n$ samples even
  when we are interested only in per-query guarantees.

  To do so, we use a key lemma proved by Indyk
  \cite{Indyk:phd,Indyk:stoc1999}.  A proof of this lemma also appears
  in \cite{Thorup04-median}, and used to establish the correctness of
  his approximate 1-median algorithm.
\begin{lemma} \label{indyk-order}
Let $Q\subset V$, $|Q|=k$ sampled uniformly at random (from all subsets of size k).
Let $u$ and $v$ be two vertices such that $\W(v) \ge (1+\epsilon) \W(u)$. Then
$\Prob(\W_Q(u) > \W_Q(v))\le e^{-\epsilon^2 |Q|/64}$.
\end{lemma}
 Lemma (\ref{indyk-order})
shows that if the average distance of two nodes differ by a factor larger than $1+\epsilon$, and we use a sample of size $\Omega( \epsilon^{-2})$ then the probability that the vertex of smaller average distance has larger average distance to the sample decays exponentially with the sample size.
This lemma immediately implies that the 1-median with respect to a sample of size $O(\log n/\epsilon^2)$
is $(1+\epsilon)$-approximate 1-median with high probability.

 To approximate all-pairs $\W(u)$, we use a uniform sample of size
 $O(\epsilon^{-2} \log n)$ and order the nodes according to the average
 distance to the sample.
Using the lemma, and comparing to the ideal sorted order by exact
$\W(v)$, two nodes $v,u$ that are transposed have with high
  probability $\W(v)$ and $\W(u)$ within $1\pm \epsilon$ from each
  other.

Recall however that the average distance to the uniform
 sample can be a very bad approximation of the average distance to the
 data set.  We therefore perform adaptively another set of
 $O(\epsilon^{-1} \log n)$ single-source distance computations to
 compute exact $\W(v)$ of enough nodes in this nearly sorted order, so
 that the difference between exact $\W(v)$ of consecutive processed
 nodes is within $(1\pm\epsilon)$.

  We also mention here, for completeness, an improved approximate
  1-median algorithm provided by  Indyk.  This algorithm only applies in metric
 spaces and computes a  $(1+\epsilon)$-approximate 1-median with constant probability
using only $O(n \epsilon^{-2})$ distance computations (eliminating the
logarithmic factor).  The algorithm works
in iterations, where in each iteration a fraction of the points, those
with largest average distance to the current sample, are
excluded from further considerations.
 The sample size is then increased by a
constant factor, obtaining more accurate estimates for the remaining points.  The final sample size used is linear, but the set of
remaining nodes is very small.  This algorithm only applies in metric
spaces because, as we mentioned in the introduction,  arbitrary distance computations are not efficient in
graphs. Indyk's approach can be extended to compute any approximate
quantile of the distribution with similar probabilistic  guarantees.

\section{Hardness of Computing Sum of All-Pairs Distances} \label{hardness:sec}

\ignore{
{\bf Shiri to fill up}

  To fill in the reduction of negative triangle to all-pairs distance
  sum.
 We create 3 copies
  of the nodes $A,B,C$.  Let $Q$ be much larger than 10 times the max
  absolute edge weight.
An edge $(u,v)$ in G has 3 edges between their copies in $A,B,C$ with
weights
$Q+w_e$ for AB and BC and $2Q-w_e$ for AC.   We extend to a complete
graph by adding edges with weight $3.5Q$.
}

In this section we show that if there is a truly subcubic algorithm for computing $\APS(V)$, the exact sum of all pairs distances then there is a truly subcubic algorithm for computing  All Pairs Shortest Paths (APSP).


Williams and Williams \cite{WilliamsW:FOCS10} showed that APSP is subcubic equivalent to negative triangle detection.
In the {\em negative triangle detection problem} we are given an undirected weighted graph $G=(V,E)$ with integer weights in $\{-M,...,M\}$ and the goal is to determine if the graph contains a negative triangle, that is, a triangle whose edge weights sum up to a negative number.
Therefore to show that a subcubic algorithm for $\APS(V)$ implies a subcubic algorithm to APSP it suffices to
give a subcubic reduction from the
negative triangle detection problem to computing $\APS(V)$.
We show this by the following lemma.

\begin{lemma}
Given a $O(T(n, m))$ time algorithm for computing the sum of all distances ($\APS(V)$) there is
$O(T(n, m)+n^2)$ time algorithm for detecting a negative triangle.
\end{lemma}
\begin{proof}
For an input instance $G=(V,E)$ for the negative triangle detection problem we construct
a graph $G' = (V',E')$ for the sum of all distances problem. The vertex set
$V'$ is the union of  three copies of $V$, that is $V' = V_1 \cup V_2 \cup V_3$ where
vertex $u_i\in V_i$, $i=1,2,3$, corresponds to vertex $u\in V$.
We set $E' = \{(u,v)\mid u,v \in V'\}$, that is $G'$ is a complete graph.

Let $\omega(e)$ denote the length of an edge $e\in E$. Recall that  $\omega(e) \in \{-M,...,M\}$.
Let $N = 4M$.
We define the length  $\omega'(e)$ of an edge $e\in E'$ as follows.
For every $(u,v) \in E$ we define
$\omega'(u_1,v_2)= N+\omega(u,v)$,
$\omega'(u_2,v_3)= N+\omega(u,v)$, and $\omega'(u_3,v_1)= 2N - \omega(u,v)$.
We set $w(e) = 3N/2$ for any other edge $e\in E'$.

We claim that $\APS(V') = \sum_{(u,v) \in E'}{\omega'(u,v)}$ if and only if $G$ does not contain a negative triangle.
In other words, we claim that either every edge in $G'$ is a shortest path or $G$ contains a negative cycle.

To see the first direction, assume $G$ contains a negative triangle $(u,v),(u,x),(x,v)$.
Now consider the path $P = (u_3,x_2), (x_2,v_1)$ from $u_3$ to $v_1$.
Note that the length of this path is $\omega'(u_3,x_2) + \omega'(x_2,v_1) = N+\omega(u_3,x_2) + N + \omega(x_2,v_1)
< 2N - \omega(u_3,v_1) = \omega'(u_3,v_1)$, where the strict inequality follows since $(u,v),(u,x),(x,v)$ is a negative triangle.
If follows that $\APS(V') < \sum_{(u,v) \in E'}{\omega'(u,v)}$.

To see the second direction, assume that $\APS(V') < \sum_{(u,v) \in E'}{\omega'(u,v)}$. We need to show that
$G$ has a negative triangle.

We first claim that for every edge $(u,v)$ which
does not correspond to an edge in $G$ (and hence $w(e) = 3N/2$)
we have $\omega'(u,v) = \dist_{G'}(u,v)$ (regardless if $G$ has a negative triangle or not).
To see this, note that $\omega'(u,v) = 3N/2 = 6M$ and that every path from $u$ to $v$ that consists of more than one edge is of weights at least
$2N-2M = 6M$. The same argument also holds for every edge from $V_1$ to $V_2$ and for every edge from $V_2$ to $V_3$.

It follows that only edges $(x,y) \in E'$ such that $x\in V_3$ and $y\in V_1$ may not be shortest paths.
If $\APS(V') < \sum_{(u,v) \in E'}{\omega'(u,v)}$ then there must be an edge $(u_3,v_1) \in E'$ such that $u_3 \in V_3$ and $v_1 \in V_1$ and the edge $(u_3,v_1)$ is not a shortest path.
It is not hard to verify that only paths of the form $(u_3,x_2), (x_2,v_1)$ such that both edges
$(u_3,x_2)$ and $(x_2,v_1)$
correspond to edges of $G$,
 could be shorter than the path $(u_3,v_1)$.
Let $(u_3,x_2),(x_2,v_1)$ be the shortest path from $u_3$ to $v_1$.
We get that $N + \omega(u_3,x_2) + N +\omega(x_2,v_1) = \omega'(u_3,x_2)+ \omega'(x_2,v_1) < \omega'(u_3,v_1) = 2N-\omega(u_3,v_1)$.
So  $\omega(u_3,x_2) + \omega(x_2,v_1) +\omega(u_3,v_1) <0$ and
$G$ has a negative triangle.
\end{proof}

\section{Extensions and Comments}

 \subsection{The distribution of centrality values}

 What can we say about the centrality distribution?
First we observe that the range of average distance $\W(v)/n$ values
is between $D/n$ to $D$, where $D$ is
the diameter (maximum distance between a pair of points in $V$).
To see the upper bound, note that the average of values that are at most $D$, is
at most $D$.   For the lower bound,
let $u$ and $v$ be nodes such that $\dist(u,v)=D$.  Then for all
$h\in V$, from triangle inequality,  $\dist(u,h)+\dist(h,v) \geq D$, thus, $\W(h) \geq D$.

 \begin{lemma} \label{highdist:lemma}
The highest average distance value  must satisfy
$$\max_{v\in V} \W(v)/n \geq D/2\ .$$
 \end{lemma}
\begin{proof}
Consider the two nodes $u$  and $v$ such that $\dist(u,v)=D$.
From triangle inequality, any point $h\in V$ has
$\dist(u,h)+\dist(h,v)\geq D$.  Summing over $h$ we obtain that
$\W(u)+\W(v) \geq nD$.  Therefore,  either
$\W(u)$ or $\W(v)$ is at least $nD/2$.
\end{proof}

\begin{lemma}
If  $z = \arg\min_{v\in V} \W(v)$ is the 1-median,  then at least half
the nodes satisfy  $\W(v) \leq 3\W(z)$.
\end{lemma}
\begin{proof}
Take the median distance $m(z)$
from $z$.  Then the average distance from $z$ is at least $m(z)/2$.  Thus, $n\cdot m(z) \leq 2\W(z)$. Consider now a node $v$ that is
one of the $n/2$ closest to $z$.  For  any node $u$,
 $\dist(v,u) \leq \dist(z,u) +m(z)$.
 Therefore, $$\W(v)=\sum_u \dist(v,u) \leq \sum_u \dist(z,u) +n m(z)
 \leq
n m(z)+\W(z) \leq 3\W(z)\ .$$
\end{proof}
Last we observe that it is easy to realize networks where there is a
large spread of centrality values.  At the extreme,  consider a single
point (node) that has distance $D$ to a very tight cluster of $n-1$
points.
The points in the cluster have $\W(v) \approx D$ whereas
the isolated point has $\W(v) \approx nD$.
More generally, networks (or data sets)  containing
well separated clusters with different sizes would exhibit a spread in centrality values.

A side comment is that as a corollary of the proof of Lemma
\ref{highdist:lemma} we obtain that
the all pairs sum
in metric spaces can be estimated with \CV\ $\epsilon$ and good concentration  by
the scaled average of distances of  $O(n\epsilon^{-2})$ pairs sampled uniformly at
random -- as established in \cite{AveDist2007}.
This is because there are at least $n-1$ pairwise
distances that are at least $D/2$, since each point that is not an
endpoint of the diameter is of distance at least $D/2$ from at least
one of the endpoints.  Since the maximum distance is $D$,
this immediately implies our claim.
 Recall, however, that
when we are restricted to using $O(\epsilon^{-2})$ single-source distance
computations from a uniform sample of nodes, the estimates can have
large \CV, but a similar bound can still be obtained using our
weighted sampling approach (see Corollary \ref{allpairssum:coro}).

 \subsection{Limitation to distances}
We showed that any set of points $V$ in any metric space can be
``sparsified''  in the sense that  a weighted sample of size
$O(\epsilon^{-2})$ allows us to estimate $\W(v)$ for any point $v$ in
the space.
We refer to such a sample as a {\em universal PPS} sample, since it
encapsulates a PPS sample of the entries in each row of the matrix.
One can ask if we can obtain similar sparsification with respect to
other nonnegative symmetric  matrices.  We first observe that in
general, 
the size of a universal PPS sample
may be $\Omega(n)$:  Consider a matrix $A_{n\times n}$ so that for
$i\in [n/2]$,
$A_{2i-1,2i} \gg 0$  but all other entries are $0$ (or close to $0$).
 The average of each row is dominated by the other member
 of the pair $(2i-1,2i)$, and therefore, any universal PPS sample must sample most
points with probability close to $1$.

 Such a matrix can not be realized with distances, as it violates the
 triangle inequality, but it can be realized when entries correspond
 to
(absolute value) of inner products of $n$ vectors in $n$-dimensional
Euclidean space  $\mathbb{R}^n$.  In this case, the sampling question we ask is 
a  well studied embedding problem \cite{TalagrandEnbed:1990}, for
which it is known that the  size of a universal PPS sample  (the
 terminology {\em leverage scores} is used) can be of size
$\Theta(d\epsilon^{-2})$, where $d$ is the dimension \cite{CohenPeng:STOC2015,TalagrandEnbed:1990}.
Intuitively, the gap between the universal PPS size between distances
and inner products stems from the
observation that being ``far'' (large distance) is something that usually applies with
respect to many nodes, whereas being ``close'' (large inner product)
is a local property.

 \subsection{Weighted centrality}
Often  different points $v$ have different importance $\beta(v)$.  In
this case, we would like our centrality measure to reflect that by
considering a weighted average of distances
 $$\frac{\sum_i \beta(i)\dist(x_i,x_j)}{\sum_i
\beta(i)}\ .$$
Our results, and in particular, the sampling construction extend to
the weighted setting.  First, instead of uniform base probabilities
$1/n$, we use PPS probabilities according to $\beta(i)/\sum_j
\beta(j)$ for node $i$.  Second, when considering
distances and probabilities from a base node, we use weight equal to the
product of $\beta(v) \dist(u,v)$ (product of $\beta$ and distance.).
Third, in the analysis, we need to
take quantiles/medians with respect to $\beta$ mass
and not just the number of points.

\subsection{Adaptive (data dependent)
  sampling} \label{adaptice1med:sec}
We showed that the number of samples needed to determine an approximate 
1-median on graphs is $O(\epsilon^{-2}\log n)$, where for each
sample we perform a single-source distance computation.
This bound is worst case which
 materializes when the
1-median $z$ is such that all other
points have $\W(u) = (1+\epsilon) \W(z)$.  In this case,   only the
exact 1-median qualifies as an approximate 1-median and also, since
there are so many other points, some are likely to have estimated
$\hat{\W}(u)< \hat{\W}(z)$ if we use a smaller sample.
On realistic instances, however, we would expect a larger separation
between the 1-median and most other points.  This 
would allow us to use fewer samples if we adaptively determine
the sample size.  Such an  approach was proposed in
\cite{binaryinfluence:CIKM2014} to identify
a node with approximate maximum marginal influence and similarly can
be applied here for the 1-median.

\section{Conclusion}
  Weighted samples are often used as compact summaries of weighted
  data.  With weighted sampling, even of very skewed data, a
  PPS sample of size $\epsilon^{-2}$ would provide us with good
  estimates with \CV\ of $O(\epsilon)$ on the total sum of the population.
   The surprise factor of our result, which relies only on
   properties of metrics, is that we can design a single set of
  sampling probabilities, which we termed {\em universal PPS},  that
  forms a good weighted sample from the
  perspectives of {\em any} point in the metric space.  Moreover, we do so in an almost
  lossless way in terms of the sample size to estimation quality
  tradeoffs.  In particular, the sample size does not depend on the
  number of points $n$ or the dimension of the space.
Another perhaps surprising consequence of our results is that there is
a rank-1 matrix that approximates the PPS
probabilities of the full pairwise distances matrix.  The approximation
can be expressed as the outer product of two vectors, which
can be computed using a linear number of distance computations.


\bibliographystyle{plain}
\bibliography{cycle}

\end{document}